\documentclass{ifacconf}
\usepackage{amsmath,amssymb,bm,bbm,mathrsfs,amscd,tikz,enumerate}
\usepackage{calc}
\usepackage{graphicx}     
\usepackage{natbib}       
\usepackage{dsfont}
\newcommand{\vect}[1]{\boldsymbol{#1}}
\newcommand{\mat}[1]{\boldsymbol{#1}}
\usepackage{subfig,pgfplots}
\pgfplotsset{
  compat=newest, 
  legend style =
  {font=\footnotesize },
  label style = {font=\footnotesize},
every tick label/.append style={font=\footnotesize}}
\newtheorem{remark}{Remark}
\newtheorem{assumption}{Assumption}
\newtheorem{problem}{Problem}
\newtheorem{example}{Example}
\def \h {1.6}
\def \l {1.8}
\def \hh {2.2}
\def \ll {2.5}
\def \hhh {1.7}
\def \lll {2.1}
\allowdisplaybreaks
\usepackage[font={footnotesize}]{caption}

\begin{document}
\begin{frontmatter}

\title{Assessing performance tradeoffs in hierarchical organizations using a diffusive coupling model\thanksref{footnoteinfo}}  
\thanks[footnoteinfo]{M. Ye is supported by the Australian Government through the Australian Research Council (DE250100199) and the Office of National Intelligence (NI240100203). }

\author[first]{Lorenzo Zino}	
\author[second]{Mengbin Ye} 
\author[third]{Brian D. O. Anderson}
\address[first]{Department of Electronics and Telecommunications, Politecnico di Torino,  10129 Torino, Italy (e-mail: lorenzo.zino@polito.it).}
\address[second]{Adelaide Data Science Centre, Adelaide University, Australia (e-mail: ben.ye@adelaide.edu.au). }
\address[third]{School of Engineering, Australian National University, Acton, ACT 2601, Australia.  (e-mail: brian.anderson@anu.edu.au).}

\begin{abstract}
We study a continuous-time dynamical system of nodes diffusively coupled over a hierarchical network to examine the efficiency and performance tradeoffs that organizations, teams, and command and control units face while achieving coordination and sharing information across layers. Specifically, after defining a network structure that captures real-world features of hierarchical organizations, we use linear systems theory and perturbation theory to characterize the rate of convergence to a consensus state, and how effectively information can propagate through the network, depending on the breadth of the organization and the strength of inter-layer communication. Interestingly, our analytical insights highlight a fundamental performance tradeoff. Namely, networks that favor fast coordination will have decreased ability to share information that is generated in the lower layers of the organization and is to be passed up the hierarchy. Numerical results validate and extend our theoretical results. \end{abstract}

\end{frontmatter}

\section{Introduction}

Many organizations across a range of sectors, including governments, large corporations, and military command and control, operate in a hierarchical network structure with multiple layers~\citep{diefenbach2011formal}. The nodes may represent people or units of people, and the layers represent different levels within the organization. Nodes may receive commands from and pass information up to the layer above, while issuing commands to and receiving information from, the layer below. Regardless of which sector one considers, key common requirements for such hierarchical organizations include ability to rapidly coordinate and reach a collective consensus state, and for information to be shared across the hierarchy~\citep{tannenbaum1962control}. 

An important question concerns the impact of the hierarchical structure (e.g. how many layers, and how many nodes per layer) and the rate at which a node passes information to the layer above or below on coordination and information spreading of the overall organization~\citep{mihm2010hierarchical,tannenbaum1962control}. Some literature has developed theories, while others have taken an empirical approach~\cite{olchi1978transmission}. Mathematical modeling offers yet another approach, in which a set of nodes is connected over a (potentially hierarchical) network structure. The nodes are connected via diffusive coupling to capture information exchange, and may have internal dynamics to reflect information processing and coordination~\citep{laguna2005dynamics,skardal2017diffusion,Fontan2018,tejedor2018diffusion,Wang2023,Zino2024,leonard2024fast}. A recent paper has considered nodes as Kuramoto oscillators to capture distributed decision-making in command and control environments; numerical simulations were used to study the model dynamics as the nonlinearity of the model limited its analytical tractability~\citep{kalloniatis2020modelling_CnC}. Another recent work has examined how the hierarchical structure could impact information spreading from the top of the hierarchy downwards, using both single and double integrator node dynamics~\citep{d2025hierarchical}.

In this paper, we further explore the performance of  organizations structured in a hierarchy, modeled as a set of nodes with single integrator dynamics connected via linear diffusive coupling. This model allows us to analytical study  how the hierarchical network structure and different strength of interactions between different layers of the organization impact its performance in terms of i) the ability to rapidly coordinate and ii) to allow information to flow across the layers. The particular structure we work with captures a hierarchical organization with multiple layers, each one formed by units, whose nodes ---representing individuals or group of individuals--- act either as a regular member of the unit or, for one node per unit,  as the unit leader. All regular members receive information from a leader of a unit in the layer below, while the leader shares and receives information from a regular member in the layer above. 

We focus on two main research questions concerning this hierarchical network model.  The first deals with the autonomous dynamics of the system. It is known that the network will converge to a consensus state~\citep{Bullo2024}.  By leveraging theoretical results on opinion dynamics models~\citep{Bullo2024}, we study the speed of convergence of the dynamics in terms of the spectral properties of the adjacency matrix of the network. For a 2-layer structure, we provide  an exact characterization of the convergence rate in terms of the model parameters, with these results being supported by numerical simulations for deeper hierarchical organizations. This provides insight into how the structure of an organization favors or hinders fast coordination, which is critical for the well-functioning of many real-world organizations.

Second, we study the problem of understanding how information flows from the bottom to the top of the network, in contrast to existing literature which typically focuses on top to bottom information flow or the autonomous dynamics. This is key for the well-functioning of a hierarchical organization, where ideally lower-layer units should be able to promptly report information to the top-layer units. For instance, in military organizations, low-layer units may gather some critical information on the enemy at the frontline, which should quickly reach the top layer. To study this scenario, we incorporate a constant nonzero input to one of the lower-layer nodes, and we study the speed at which the network converges to the input. Using again spectral properties of the adjacency matrix and a perturbation argument~\citep{Greenbaum2020}, we illustrate analytically the existence of a fundamental tradeoff, between how fast the hierarchical network can coordinate and how fast information can spread bottom-up. This suggests that a hierarchical organization aiming to be efficient in both tasks should carefully balance  information processing across the layers. 

In summary, the main contributions of this paper are threefold. First, we propose a hierarchical network model that captures salient features of real-world organizations, with a parsimonious formalization that is amenable to analytical treatment. Second, we establish analytical insights into how the hierarchical structure impacts the speed at which coordination emerges in the organization via consensus formation. Third, we identified a fundamental tradeoff between the ability for information to diffuse rapidly up the hierarchy and achieving fast coordination across the entire organization. 

The rest of the paper is organized as follows. In Section~\ref{sec:network}, we formalize and discuss the hierarchical network structure. In Section~\ref{sec:problem}, we formulate the research problem. In Section~\ref{sec:results}, we present our main results. Section~\ref{sec:conclusion} concludes the paper and outlines future research. 

\section{Hierarchical network model}\label{sec:network}

\subsubsection{Notation}
We gather here the notational conventions used throughout the paper. We denote as $\mathbb R$ and $\mathbb R^+$ the sets of real and nonnegative real numbers, respectively. Given a vector $\vect{x}\in\mathbb R^n$ (or a matrix $\mat{M}\in\mathbb R^{m\times n}$), we denote by $\vect{x}^\top$ and ($\mat M^\top$) its transpose, and by $\text{diag}(\vect{x})\in\mathbb R^{n\times n}$ the diagonal matrix with the components of $\vect{x}$ on the diagonal. Given $M\in\mathbb Z_+$, $\vect{1}_M$ is the $M$-dimensional vector of all $1$s, $\vect{0}_M$ is the $M$-dimensional vector of all $0$s, $I_M$ is the $M$-dimensional identity matrix, and $e_k^{(M)}$ is the $M$-dimensional vector of all $0$s except for a $1$ in the $k$th component.

\subsection{Hierarchical network structure}
In this paper, we consider a network $\mathcal{G} = (\mathcal V, \mathcal E)$, where the set $\mathcal{V} = \{1,2,\hdots, n\}$ represents $n$ nodes and $\mathcal{E}$ is the directed edge set, with an edge $(i,j) \in \mathcal E$ originating from node $i$ and terminating at node $j$. Nodes represent actors in the organization (individuals, or aggregate groups acting as one) and edges capture communication between two actors.


We consider a specific hierarchical network structure that aims to provide a general description of many organizations, acknowledging in reducing the number of parameters needed to describe the network that we overlook some heterogeneity that exists for organizations in different contexts and domains. 

The network is made of $L\geq 2$ \emph{layers}, labeled in decreasing hierarchical order: the uppermost layer is layer $1$, the lowest is layer $L$. Each layer of the network corresponds to a layer in the organization. The basic element of this hierarchical structure is called a \emph{unit}. The structure comprises a single unit in the uppermost layer, and multiple units in the other layers, as described below. We assume that all units are of the same size, containing $M+1$ nodes, with $M\geq 2$. Within a unit, the $M+1$ nodes form a fully connected clique, i.e., each and every node interacts with every other node. In each unit there is a `unit leader' and $M$ `members'. The unit leader with the exception of the leader of the unit in the uppermost layer is responsible for communicating with precisely one member of the next upper layer of the hierarchy, while each of the other $M$ members except those in a unit in layer $L$ is responsible for communicating with the leader of precisely one unit at the next lower layer of the hierarchy. Hence, corresponding to each unit at level $\ell<L$, there are $M$ (sub-)units at level $\ell+1$, as illustrated in Fig.~\ref{fig:net2}.

\begin{figure}
    \centering
    \includegraphics[width=\linewidth]{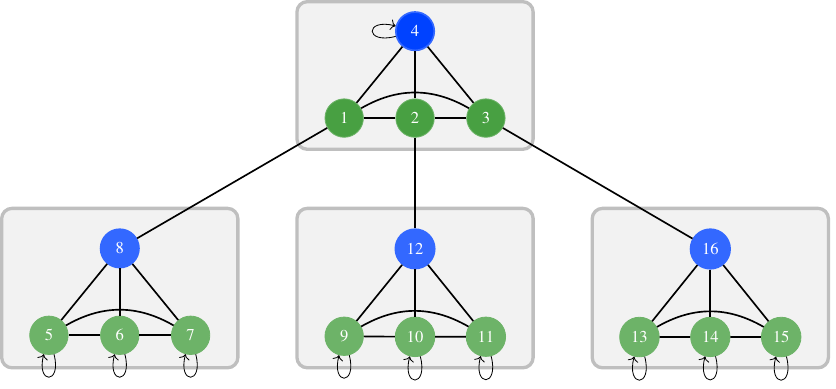}
    \caption{Representation of the hierarchical network with $L=2$ and $M=3$. Units are represented by gray boxes. Leaders and members of the units are represented in shades of blue and green, respectively.}
    \label{fig:net2}
\end{figure}

Given $L$ and $M$, the hierarchical network has $N=\frac{(M+1)(M^L-1)}{M-1}$ nodes. For convention, we assign indices to nodes unit-by-unit, starting from the top layer, labeling members before the leaders, so that nodes $1,\dots,M$  are the member of the uppermost unit and $M+1$ is its leader;  nodes $M+2,\dots,2M+1$  are the member of the first unit of the second layer and $2M+2$ is its leader; etc. 

\subsection{Weighted adjacency matrix}

In this paper, we consider information flowing over the  hierarchical network structure described in the above. To capture this, we introduce a weighted adjacency matrix $\mat{W}\in[0,1]^{n\times n}$. In particular, we assume that the weight given to information depends only on the relative levels of the layers in which the nodes exchanging information are located in the hierarchical structure. Specifically, we define two parameters $\alpha> 0$ and $\beta> 0$, which capture the value of information coming to any node of a given unit from the layer above and layer below, respectively, relative to the information coming  from the other members of the same unit. For the sake of simplicity we normalize weights to obtain a stochastic weight matrix and we also assume that the leader of the first unit and the members of the last-layer units are connected with a self-loop with a weight that guarantees stochasticity (i.e., such that $\mat{W}\vect{1}=\vect{1}$). We summarize these assumptions in the following. 

\begin{assumption}\label{a:network}
Let us write the index of a node as $i=i_3+(i_2-1)(M+1)+(M+1)\frac{M^{(i_1-1)}-1}{M-1}$, with $i_1$, $i_2$, and $i_3$ being the layer, the unit, and the position in the unit, respectively. Given two positive constants $\alpha>0$ and $\beta>0$, the weight matrix $\mat W$ is defined as
\begin{equation}
    w_{ij}=\left\{\begin{array}{cl}
    \tfrac{1}{M+\beta}&\text{if }i_1=j_1, i_2=j_2, i_3\neq M+1;\\
        \tfrac{1}{M+\alpha}&\text{if }i_1=j_1\neq 1, i_2=j_2, i_3= M+1;\\
    \tfrac{\beta}{M+\beta}&\text{if }j_1=i_1+1, j_2=M(i_2-1)+i_3,\\&\quad j_3=M+1;\\
    &\text{or }i_1=j_1=L, i_2=j_2, i_3=j_3\neq M+1;\\
    \tfrac{\alpha}{M+\alpha}&\text{if }i_1=j_1+1, i_2=M(j_2-1)+j_3,\\&\quad i_3=M+1; \\  
    &\text{or }i_1=i_2=j_1=j_2=1, i_3=j_3=M\hspace{-.05cm}+\hspace{-.05cm}1;\\
    0&\text{otherwise}.\end{array}\right.
\end{equation}

\end{assumption}

\begin{remark}\label{rem:symmetry}
In general, the graph $\mathcal G$ is bidirectional, but the weight matrix $\mat W$ is  non-symmetric (i.e., $\mat W\neq\mat W^\top$), with the exception of the special case in which $\alpha=\beta$.
\end{remark}

The weighted network structure described in Assumption~\ref{a:network} allows us to capture some key properties of hierarchical organizations. While the parameters $L$ and $M$ represent the depth and breadth of the structure, parameters $\alpha$ and $\beta$ encapsulate the heterogeneity that is typically present in how information is processed by the members of a unit, depending on whether the source of such information comes from the layer above or below. These two parameters compactly capture the organizational policies in terms of information sharing across the different levels of the hierarchy and importance given to such information. In the rest of the paper, we will focus our analysis on understanding how such parameters impact the performance of the organization. Before moving to that analysis, we present in more detail an example with $L=2$, which will be used in the following.

\begin{example}[Case $L=2$]\label{ex:L2}
Of particular interest is the case $L=2$, which is the simplest scenario in which the network has a nontrivial layered structure. In this case, the network consists of one unit in the first layer with leader and $M$ other nodes, and $M$ units in the second layer, for a total of $(M+1)^2$ nodes, as illustrated in Fig.~\ref{fig:net2} for $M=3$.  
The matrix $\mat W$ under Assumption~\ref{a:network} can be written in a compact form as\begin{equation}\label{eq:W}
   \mat  W=\begin{bmatrix}\mat P+  \mat E_{\alpha,M+1}&\mat E_{\beta,1}&\mat E_{\beta,2}&\dots&\mat E_{\beta,M}\\
 \mat  E_{\alpha,1}^\top& \mat P+\mat Q&\vect{0}\vect{0}^\top&\dots&\vect{0}\vect{0}^\top\\
       \mat E_{\alpha,2}^\top& \vect{0}\vect{0}^\top&\mat P+\mat Q&\dots&\vect{0}\vect{0}^\top\\\vdots& \vdots&&\ddots&\vdots\\
         \mat E_{\alpha,M}^\top& \vect{0}\vect{0}^\top&\dots&\vect{0}\vect{0}^\top&\mat P+\mat Q
    \end{bmatrix},
\end{equation}
with 
\begin{equation}\begin{array}{rcl}
    \mat P&=&\text{diag}([\tfrac{1}{M+\beta},\dots,\tfrac{1}{M+\beta},\tfrac{1}{M+\alpha}])(\vect{1}\vect{1}^\top- I),\\
    \mat Q&=&\tfrac{\beta}{M+\beta}( I-e_{M+1}e_{M+1}^\top),\\
    \mat E_{x,i}&=&\tfrac{x}{M+x}e_ie_{M+1}^\top,\,\,\,\mbox{ for }x=\alpha,\beta.
\end{array}
\end{equation}
where the dimension of all vectors and matrices involved is equal to $M+1$ (i.e., we use $\vect{1}$ for $\vect{1}_{M+1}$, $\vect{0}$ for $\vect{0}_{M+1}$, $I$ for $I_{M+1}$, and $e_{i}$ for $e_i^{(M+1)}$). 
\end{example}

\section{Dynamics and Problem Formulation}\label{sec:problem}

\subsection{Dynamics}

In this paper, we are interested in understanding the impact of the hierarchical structure on the performance of an organization in terms of coordination and information flowing. To investigate this question reducing all possible external confounding, we consider the simplest diffusive coupling model, in which all nodes have no internal dynamics and act as single integrators of the relative pairwise differences and, possibly, of an external input. 

In particular, each node $i\in\mathcal V$ is associated with a scalar variable $x_i(t)\in\mathbb R$, representing the state of the node. This scalar variable evolves in continuous  time ($t\in\mathbb R)$ according to the following dynamics:
\begin{equation}\label{eq:dyn}
    \dot x_i=\sum_{j\in\mathcal V}w_{ij}(x_j-x_i)+\gamma_i(u_i(t)-x_i),
\end{equation}
where $u_i(t)$ is a (locally integrable) external input, and the constant non-negative parameter $\gamma_i\in\mathbb R_+$ captures the intensity of the input, whereby $\gamma_i=0$ means that no input is exerted at node $i$ and $\gamma_i>0$ means that an input is exerted at node $i$. We say that the system has a single input if the vector $\vect{\gamma}$ comprises all zeros, except for one strictly positive entry.

This single-integrator model, whose dynamics can be written compactly as $\vect{\dot x}=\mat A\vect{x}+\text{diag}(\vect{\gamma})u$, with $\mat A= -(I_n-\mat W+\text{diag}(\vect\gamma))$, give rise to the well-known consensus dynamics, which has been extensively used to study social influence and opinion formation in complex social systems, see e.g.~\cite{Proskurnikov2017}.  In fact, \eqref{eq:dyn} captures information exchange between nodes and the tendency to compromise in order to reach a consensus. Hence, the parameter $w_{ij}$ captures the rate at which node $i$ uses information from $j$. The larger this quantity, the faster node $i$ tends to adjust its variable $x_i(t)$ to approach $x_j(t)$. All model and network parameters are summarized in Table~\ref{tab:parameters}.

\begin{table}
\centering
\caption{Model variables and parameters.}\label{tab:parameters}
\begin{tabular}{r| l}
$n$& population of the network\\
$x_i(t)$& state of node $i$ at time $t$\\
$u_i(t)$& input value at node $i$\\
$\gamma_i$& input intensity at node $i$\\
$L$& number of layers (depth)\\
$M$& number of sub-units associated to each unit (breadth)\\
$\alpha$& value of information coming from the layer above\\
$\beta$& value of information coming from the layer below
\end{tabular}
\end{table} 

We report the following results concerning the asymptotic behavior of \eqref{eq:dyn}, which will be used in our analyses. The proofs are based on classical theory of linear systems specialized to dynamical flows~\citep{Bullo2024}, and are thus omitted due to space limitations.
\begin{prop}\label{prop:convergence}
     In the absence of any input ($\gamma_i=0$ for all $i\in\mathcal V$), \eqref{eq:dyn} converges to a consensus state, i.e., $\lim_{t\to\infty}\vect{x}(t)=\bar x\vect{1}_n$. In particular, $\bar x=\vect{\pi}^\top\vect{x}(0)$, where $\vect{\pi}$ is the (normalized) left eigenvalue associated with the dominant eigenvalue of $\mat{W}$. Moreover,
\begin{equation}\label{eq:rate}
\|\vect{x}(t)-{\bar x}\vect{1}\|\leq K\exp\{-rt\},    
\end{equation}
where $r$ is the modulus of the second largest eigenvalue of $\mat W$ in modulus.
\end{prop}

\begin{prop}\label{prop:convergence2}
     In the presence of a single constant input (single entry $\gamma_i>0$ with $u_i(t)=u$), \eqref{eq:dyn} converges to the consensus state $\lim_{t\to\infty}\vect{x}(t)=u\vect{1}_n$, Further, with $\lambda_A$ defined as the largest eigenvalue of $\mat W-\text{diag}(\vect\gamma)$, then   \eqref{eq:rate}
holds where $r=1-\lambda_A$.
\end{prop}

In the rest of this paper, we will use Propositions~\ref{prop:convergence} and~\ref{prop:convergence2} to unveil how the hierarchical network structure shapes the consensus state and the speed of convergence to it.

\subsection{Problem formulation}


In order to formalize our research questions, we start by presenting a motivating example, which illustrates how designing a hierarchical structure for an organization is a nontrivial problem due to a tradeoff that seem to inherently emerge between how fast a hierarchical structure can reach a common agreement (i.e., a consensus) and how fast information can flow across the network. This tradeoff is critical in real-world hierarchical organizations. In military organization, e.g., one would generally like all individual troops to be able to promptly align to the order received, but also for reports from the lower-level units (e.g., who may obtain critical information on the enemy) to quickly be received by the top layers of the organization. 

\begin{example}[Motivating example]\label{ex:example}
Simulations in Fig.~\ref{fig:example}  illustrate two different instances of \eqref{eq:dyn}  on a hierarchical network with $L=2$ and $M=3$ (see Fig.~\ref{fig:net2}). In Figs.~\ref{fig:example1}--\ref{fig:example3}, we consider the system with no input, with a random initial condition (equal in all simulations). In the simulations, we fix $\beta=1$ and change the value of $\alpha$. Increasing $\alpha$ seems always beneficial in speeding up the consensus formation process, shaping its consensus state. In Figs.~\ref{fig:example4}--\ref{fig:example6}, we introduce an input $u=1$ to node $i=5$ (a member of a second-layer unit), setting all initial conditions to $\vect{x}(0)=\vect{1}$. Interestingly, the speed of convergence to the input is non-monotone: as $\alpha$ increases, an optimum is reached, then the performance seems to get worse. Note that this effect is observed not only when considering convergence of the whole network, but also when considering the leading unit (solid trajectories, shades of blue).
\end{example}

\begin{figure}
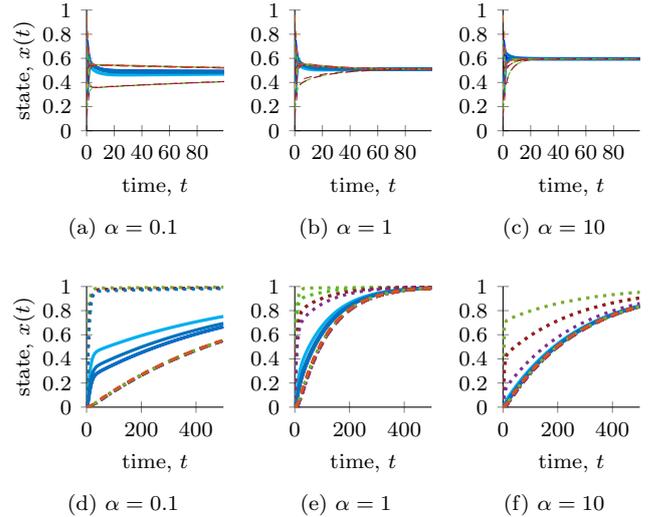

    \centering
\subfloat[$\alpha=0.1$]{\input{fig/a01}    \label{fig:example1}}
\subfloat[$\alpha=1$]{\input{fig/a1}    \label{fig:example2}}
\subfloat[$\alpha=10$]{\input{fig/a10}    \label{fig:example3}}\\
\subfloat[$\alpha=0.1$]{\input{fig/b01}    \label{fig:example4}}
\subfloat[$\alpha=1$]{\input{fig/b1}    \label{fig:example5}}
\subfloat[$\alpha=10$]{\input{fig/b10}    \label{fig:example6}}
\caption{Dynamics with (a--c) autonomous consensus and (d--f) bottom-up consensus with different values of $\alpha$. Solid curves correspond to the nodes in the unit on the first layer, dotted curves are nodes in the unit with the input (with the green curve the one of the node with the input), dashed curves are nodes in the other unit on the second layer. In (a--c), since no input is present, all trajectories on the second layer are dashed. Common parameters are $L=2$ and $M=3$. }
    \label{fig:example}
\end{figure}

Motivated by Example~\ref{ex:example}, we focus on two different aspects related to the efficiency of a hierarchical network structure, viz. \textit{autonomous consensus} and \textit{bottom-up consensus}. In the first problem, we are interested in understanding how fast coordination spontaneously emerges in the network. To study this scenario, we assume that the system has no input, and we study how the different value given to information coming from different layers and the network breadth shape the consensus state reached and the rate at which it emerges, i.e., the time needed to reach coordination across the organization. Then, we study how information flows from the bottom of the organization to the top. To study this scenario, we place an input to a member of the lowermost layer, and we study how fast the network converges to this input. We analytically study these two problems for hierarchical networks with $L=2$ layers,  illustrating numerically how our results are coherently extended (and often emphasized) for deeper networks. We formalize these two problems as follows.

\begin{problem}[Autonomous consensus]\label{p:1}
    Assume that $L=2$ and $\gamma_i=0$ for all $i\in\mathcal V$. Determine the consensus state of \eqref{eq:dyn} and its rate of convergence as functions of $M$, $\alpha$, and $\beta$.
\end{problem}

\begin{problem}[Bottom-up consensus]\label{p:2}
    Assume that $L=2$, $\gamma_{M+2}>0$, $u_{M+2}(t)=u$, and $\gamma_i=0$ for all $i\neq M+2$. Determine the rate of convergence of  \eqref{eq:dyn} to the input as a function of $M$, $\alpha$, and $\beta$.
\end{problem}

\section{Results}\label{sec:results}

\subsection{Problem I: Autonomous consensus}

In this section, we study Problem~\ref{p:1}  (autonomous consensus), where the system in \eqref{eq:dyn} has no input. Motivated by the observations from Example~\ref{ex:example}, we are interested in understanding how the hierarchical structure shapes the final consensus reached and the speed of convergence. 

In this regard, Proposition~\ref{prop:convergence} guarantees that the system converges to a consensus state, whose value is shaped by the the left eigenvector associated with the dominant eigenvalue of $\mat W$ and the speed of convergence is determined by the second largest eigenvalue. 
Hence, Problem~\ref{p:1} ultimately reduces to the study of the spectrum of $\mat W$. For $L=2$, despite the complexity due to the non-trivial structure and the non-symmetry of the matrix (see Remark~\ref{rem:symmetry}, the regularity in the structure allows us to explicitly compute the eigenvalues and eigenvectors, as summarized in the following propositions.

 
\begin{prop}\label{prop:pi}
The left eigenvector associated with the largest eigenvalue $\lambda_A=1$ of  $W$ in \eqref{eq:W} is
    \begin{equation}\label{eq:left}
        \vect{\pi}=K\begin{bmatrix}
            \vect{1}_M^\top&
            \tfrac{M+\alpha}{M+\beta}&
            \tfrac{\beta}{\alpha}\vect{1}_M^\top&
            \tfrac{(M+\alpha)\beta}{(M+\beta)\alpha}&
                \dots&    \tfrac{\beta}{\alpha}\vect{1}_M^\top&
            \tfrac{(M+\alpha)\beta}{(M+\beta)\alpha}&
        \end{bmatrix}^\top,
    \end{equation}
    with normalization factor $K=\frac{\alpha(M+\beta)}{(M^2+M+\alpha+M\beta)(M\beta+\alpha)}$.
\end{prop}
\begin{pf}
    The proof follows from checking that $\vect{\pi}^\top \mat W=\vect{\pi}^\top$.
\end{pf}

\begin{prop}\label{prop:spectrum_alphabeta}
The matrix $W$ defined in \eqref{eq:W} has all real eigenvalues: the dominant eigenvalue $\lambda_A=1$ has multiplicity $1$; eigenvalues
\begin{equation}\label{eq:lambdaB_general}
     \begin{array}{ll}\lambda_{B,C,D}&=\frac{\beta K_{B,C,D}(M+\beta)-\beta K_{B,C,D}+1}{(M+\beta)(\beta K_{B,C,D}-1)}\\
     &=
    1-\frac{\beta K_{B,C,D}-M-\beta-1}{(M+\beta)(\beta K_{B,C,D}-1)},
    \end{array}
 \end{equation}
 each with multiplicity $M-1$, where $K_{B}>K_{C}>K_{D}$ are the three real solutions of the third-order equation
 \begin{equation}\label{eq:K_general}
(\alpha+MK)(\beta K-1)^2-(M+\alpha)K(\beta K(M+\beta)-\beta K+1)=0;
 \end{equation}
eigenvalue $\lambda_E=\frac{\beta-1}{M+\beta}$ with multiplicity $M(M-1)$, and the three single eigenvalues $\lambda_F=\frac{M(\alpha-1)+\alpha (\beta-1)}{(M+\alpha)(M+\beta)}$ and
\begin{equation}\label{eq:gh}
    \lambda_{G,H}
=\frac{M-1\pm\sqrt{(M-1)^2+\frac{4(M+\beta)(M+\alpha\beta)}{M+\alpha}}}{2(M+\beta)}.
\end{equation}
\end{prop}
\begin{pf}
The proof is provided in Appendix~\ref{app:proof1}.
\end{pf}

It is worth noticing that \eqref{eq:K_general} is a third-order equation, so we can derive a closed-form expression for all the eigenvalues of $W$. Due to its cumbersomeness, we omit reporting of such an expression. 

Finally, building on Propositions~\ref{prop:convergence},~\ref{prop:pi}, and~\ref{prop:spectrum_alphabeta}, we can establish the following result, which solves Problem~\ref{p:1}, characterizing the consensus state and the speed of convergence for the autonomous consensus problem on a hierarchical network with $L=2$.


\begin{thm}\label{th:problem1}
The state $\vect{x}(t)$ converges to a consensus $\vect{\bar x}=\vect{\pi}^\top \vect{x}(0)\vect{1}$ with $\vect{\pi}$ from \eqref{eq:left}, and it holds
\begin{equation}
\|\vect{x}(t)-\vect{\bar x}\|\leq K\exp\{-rt\},    
\end{equation}
with 
\begin{equation}\label{eq:r}
     r=1-\max\{\lambda_B,\lambda_G\}
 \end{equation}
from \eqref{eq:lambdaB_general} and \eqref{eq:K_general}. 
 \end{thm}
 \begin{pf}The results comes from specializing Proposition~\ref{prop:convergence} using the results of Propositions~\ref{prop:pi} and \ref{prop:spectrum_alphabeta}, after proving that $\lambda_B$ or $\lambda_G$ is the second largest eigenvalue of $\mat W$. 
 
 First, we prove that $\lambda_{B}>\lambda_{D}>\lambda_{C}$. To this aim, we study \eqref{eq:lambdaB_general} as function of $K$, i.e., $\lambda(K)$, from which we derive $\lambda_B=\lambda(K_B)$, $\lambda_C=\lambda(K_C)$, and $\lambda_D=\lambda(K_D)$. The function is a rectangular hyperbola with vertical asymptote at $K^*=\frac{1}{\beta}$. Hence, $\lambda(K)$ is locally monotonically decreasing in each side of the hyperbola, i.e., in the two intervals $K\in(-\infty,\frac{1}{\beta})$ and $K\in(\frac{1}{\beta},\infty)$. Moreover, the horizontal asymptote of the hyperbola is $\frac{M+\beta-1}{M+\beta}$. Hence, $\lambda(K)\geq \frac{M+\beta-1}{M+\beta}$ for any  $K<\frac1\beta$ and $\lambda(K)\leq\frac{M+\beta-1}{M+\beta}$ for any  $K>\frac1\beta$. Finally, we observe that  $K^*=\frac{1}{\beta}$
is between $K_B$ and $K_C$, this implies that $K_B>\frac{M+\beta-1}{M+\beta}$, while $K_C<\frac{M+\beta-1}{M+\beta}$ and $K_D<\frac{M+\beta-1}{M+\beta}$. Finally, monotonicity of $\lambda(K)$ yields the ordering  $\lambda_{B}>\lambda_{D}>\lambda_{C}$.  

From the bound $K_B>\frac{M+\beta-1}{M+\beta}$, we can also immediately verify that $\lambda_B>\lambda_E$ (being $M+\beta-1>\beta-1$). Then, we also prove  that $\lambda_B>\lambda_F$. In fact, since $\lambda_B>\frac{M+\beta-1}{M+\beta}$, we have just to verify that $M+\beta-1>M(\alpha-1)+\alpha(\beta-1)$, which ultimately simplifies to the condition $M(M+\beta)>0$, which is always satisfied. Finally, observing that $\lambda_G>\lambda_H$, we conclude that the second largest eigenvalue is either $\lambda_B$ or $\lambda_G$, yielding \eqref{eq:r}. \qed
 \end{pf}

     The question of which of $\lambda_B$ and $\lambda_G$ is the larger to determine \eqref{eq:r} ultimately depends on the model parameters. Typically, it  occurs that $\lambda_B>\lambda_G$ (see, e.g., Corollary~\ref{cor:equal}), except for scenarios in which $M$ is small, $\alpha$  large and $\lambda_B$ small, as suggested by the plots in Fig.~\ref{fig:lambda}, which show that the region in which $\lambda_G>\lambda_B$ shrinks as $M$ grows.

     \begin{figure}
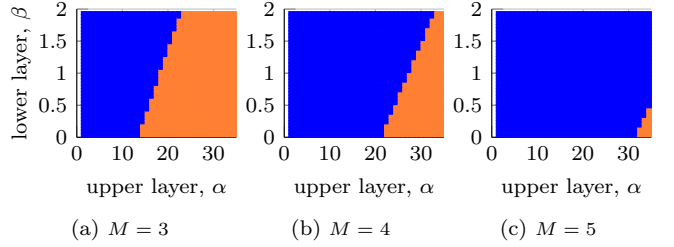

         \centering
         \subfloat[\scriptsize$M=3$]{\input{fig/3a}\label{fig:lambda_1}}\subfloat[\scriptsize$M=4$]{\input{fig/3b}\label{fig:lambda_2}}\subfloat[\scriptsize$M=5$]{\input{fig/3c}\label{fig:lambda_3}}
         \caption{Comparison between $\lambda_B$ and $\lambda_G$ computed using \eqref{eq:lambdaB_general} and \eqref{eq:gh}, respectively, for different breadths $M$. In the blue region $\lambda_B>\lambda_G$, in the orange region $\lambda_G>\lambda_B$.}
         \label{fig:lambda}
     \end{figure}

 Interestingly, we can make the following observation on the dependency of the convergence rate on the model parameters, which provides analytical support to the observations from Example~\ref{ex:example}.
 
\begin{cor}\label{cor:rate}
     The rate $r$ from \eqref{eq:r} is monotonically increasing with respect to $\alpha$.
 \end{cor}
\begin{pf}If the second largest eigenvalue is $\lambda_G$, then the claim can be immediately derived by checking that only the last term in the square root of \eqref{eq:gh} depends on $\alpha$, and it is monotonically decreasing in $\alpha$. Hence, $1-\lambda_{G}$ is monotonically increasing in $\alpha$. If the second largest eigenvalue is $\lambda_B$, we  evaluate its dependence on $\alpha$ by means of the total derivative \begin{equation}\label{eq:total}
    \frac{d\lambda_B}{d\alpha}=\frac{\partial\lambda_B}{\partial\alpha}+\frac{\partial\lambda_B}{\partial K_B}\frac{d K_B}{d\alpha}. 
    \end{equation}From \eqref{eq:lambdaB_general}, we observe that $\frac{\partial \lambda_B}{\partial \alpha}=0$ and we compute $\frac{\partial \lambda_B}{\partial K}=-\frac{\beta}{(\beta K-1)^2}<0$, for any $K\neq \frac1\beta$. Since we know that $K_B$ is the largest solution of \eqref{eq:K_general}, for which there holds $K_B>\frac1\beta$, in order to 
guarantee that $\frac{d\lambda_B}{d\alpha}<0$, according to the observations made above on the first two terms of \eqref{eq:total}, we are only left to demonstrate that  $\frac{d K_B}{d\alpha}>0$, i.e., that $K_B$ is monotonically increasing with respect to $\alpha$. To show that the largest solution of \eqref{eq:K_general} is monotonically increasing with respect to $\alpha$, we denote the left-hand side of \eqref{eq:K_general} as $F(K)$. We observe that
        $\frac{\partial F}{\partial \alpha}=(\beta K-1)^2-K(\beta K(M+\beta)-\beta K+1)=1-2\beta K-\beta (M-1) K^2-K<-1-\frac{M}{\beta}<0$,   for any $K>\frac{1}{\beta}$. Finally, we observe that $F(K)$ is a cubic function that tends to $+\infty$ as $K\to+\infty$. Hence, the fact that $\frac{\partial F}{\partial \alpha}<0$ for any $K>\frac{1}{\beta}$ combined with the fact that the largest solution of $F(K)=0$ is greater than $\frac{1}{\beta}$ implies that such a solution, i.e., $K_B$ is monotonically increasing in $\alpha$, yielding the  claim. 
    \qed
\end{pf}

Corollary~\ref{cor:rate} depicts a clear monotonic dependence of the convergence rate on $\alpha$ for the autonomous consensus problem: to reach fast synchronization it is always beneficial to increase the value given to information coming from the upper layer. Intuitively, this quickly steers the whole hierarchical structure to reach a consensus close to the initial opinion of the leader of the first unit. In fact, from Proposition~\ref{prop:pi}, we observe that the corresponding weight (viz. the $(M+1)$th entry of ${\vect\pi}$) becomes dominant as $\alpha$ grows. 

Figure~\ref{fig:problem1_1} illustrates such a  monotone dependence, supporting the analytical insights from Corollary~\ref{cor:rate}, and providing further insights into the dependence on $\beta$, which appears instead to be monotonically decreasing. In the figure, darker shades of green are associated with larger convergence rates. Interestingly, we observe that such dependence shows great differences in sensitivity across the parameter space, whereby a small increase of $\beta$ seems to be able to jeopardize the beneficial effect of large values of $\alpha$. Finally, it is worth noticing that, while our analysis is limited to the case $L=2$, the numerical computation of the second largest eigenvalue for $L=3$ seems to yield results that are consistent with our analytical findings for $L=2$ (see Fig.~\ref{fig:problem1_2}, where the presence of an additional layer seems to exacerbate  the differences in sensitivity  described above.

 \begin{figure}
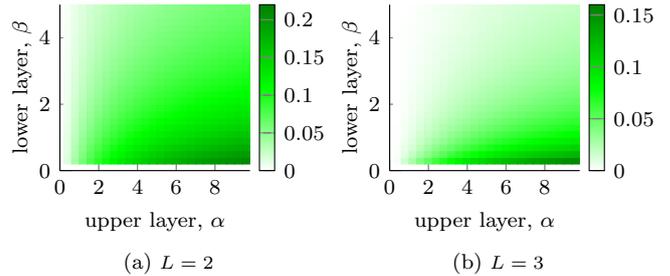

     \centering
\subfloat[\scriptsize$L=2$]{\input{fig/4a}\label{fig:problem1_1}}
\subfloat[\scriptsize$L=3$]{\input{fig/4b}\label{fig:problem1_2}}
     \caption{Convergence rate $r$ for the autonomous consensus problem, computed (a) analytically  
 for $L=2$ and  (b) numerically for $L=3$, for different values of $\alpha$ and $\beta$. In both scenarios $M=3$. }
     \label{fig:problem1}
 \end{figure}

 Finally, it is worth noticing that, when we constrain information coming from different layers to have the same weight (i.e., $\alpha=\beta$), the results in Proposition~\ref{prop:spectrum_alphabeta} simplify, yielding an easy closed-form expression of the convergence rate, which is reported in the following corollary. 
 \begin{cor}\label{cor:equal}
If $\alpha=\beta$, then $\vect{\pi}=\frac{1}{N}\vect1$, and \eqref{eq:r} simplifies to 
\begin{equation}\label{eq:r_same}
 r=
 \frac{M+2\alpha+1-\sqrt{(M+2\alpha+1)^2-4\alpha}}{2(M+\alpha)}.
\end{equation}
 \end{cor}
 \begin{pf}
     When $\alpha=\beta$, $\vect\pi$ simplifies to a uniform vector.  Moreover, \eqref{eq:K_general} has a solution for $K_D=-\frac{1}{M}$, yielding $\lambda_D=\frac{\alpha-1}{M+\alpha}$. This simplifies the computation of the other two solutions of the equation, yielding \begin{equation}\label{eq:b_same}
 \lambda_B=\frac{M-1+\sqrt{(M+2\alpha+1)^2-4\alpha}}{2(M+\alpha)}.
\end{equation}   
     Finally, it is possible to directly compare $\lambda_B$ in \eqref{eq:b_same} and $\lambda_G$ in \eqref{eq:gh} with $\alpha=\beta$. In fact, we observe that $(M+2\alpha+1)^2-4\alpha\geq (M+2\beta-1)^2+4\alpha+\frac{4M(M-\alpha)(1-\alpha)}{M+\alpha}$, implying that $\lambda_B>\lambda_G$, which yields the claim.\qed
 \end{pf}

Interestingly, when $\alpha=\beta$, i.e., when the members of a unit give the same value to the information received from the upper and lower layers of the hierarchy, the dependence on $\alpha$ becomes nontrivial. In fact, from \eqref{eq:r_same}, we observe that in the two limits $\lim_{\alpha\to 0}r=\lim_{\alpha\to \infty}r=0$, i.e., convergence becomes slow. From its analytical expression in \eqref{eq:r_same}, we can compute the optimal value of $\alpha=\beta$ to speed up the emergence of a consensus, as the argument that maximizes \eqref{eq:r_same}, obtaining \begin{equation}
  \alpha^*=\beta^*=\frac{M+(M-1)\sqrt{M(2M+1)}}{2M-1}.
\end{equation}
Interestingly, the optimal value of $\alpha=\beta$ grows with $M$, ultimately approaching $\alpha^*= \frac{\sqrt 2}{2}M+o(M)$. To better understand such a linear dependence on $M$, we observe that, in order to reach fast convergence to a consensus, the weight of the information coming from the other layers should compensate for the fact that the size of each unit grows linearly with $M$, and so does the total weight of the interactions within the unit.



\subsection{Problem II: Bottom-up consensus}

Here, we study the bottom-up consensus problem, in which a non-zero input is introduced in a member of one of the bottom-layer units, and we are interested in understanding how fast the entire network converges to such an input. Example~\ref{ex:example} suggests that the problem is non-trivial, as the rate of convergence to the input seems to depend on the model parameters in a non-monotone manner (especially on $\alpha$). Here, we offer some analytical insights supporting this claim. First, from Corollary~\ref{cor:rate}, we observe that in the absence of any input, the rate of convergence is monotonically increasing with respect to $\alpha$. Second, we consider Proposition~\ref{prop:convergence2}, which relates the rate of convergence to the input with  the dominant eigenvalue of $\mat{\tilde W}=\mat W - \gamma e_{M+2}e_{M+2}^\top$, and we observe the following.

\begin{prop}\label{prop:input}
The largest eigenvalue of $\tilde {\bf{W}}$ is equal to
    \begin{equation}\label{eq:linearA}
    \lambda_A(\tilde{\bf{W}})=-\gamma\frac{\beta(M+\beta)}{(\beta M+\alpha)(M^2+M\beta+M+\alpha)} +o(\gamma).
\end{equation} 
\end{prop}
\begin{pf}
By applying an eigenvalue perturbation argument~\cite{Greenbaum2020}, we observe that the largest eigenvalue of $\mat{\tilde W}=\mat W-\gamma e_{M+2}e_{M+2}^\top$ can be written as
\begin{equation}\label{eq:perturbation}\begin{array}{ll}
    \lambda_A(\tilde{\bf{W}})&=\lambda_A({\bf{W}})-\gamma\frac{1}{\vect{\pi}^\top\vect{v}} \vect{\pi}^\top e_{M+2}e_{M+2}^\top \vect{v}+o(\gamma)\\&=1-\gamma\frac{1}{\vect{\pi}^\top\vect{v}} {\pi}_{M+2}{v}_{M+2}+o(\gamma),
\end{array}
\end{equation} 
where $\vect{v}$ and $\vect{\pi}$ are the right and left eigenvectors of ${\bf{W}}$ associated with $\lambda_A=1$, which are known from the proof of Proposition~\ref{prop:spectrum_alphabeta} ($\vect{v}=\vect{1}$) and from Proposition~\ref{prop:pi}. Hence, substituting the corresponding entries and recalling that $\vect\pi$ is normalized (so $\vect\pi^\top\vect{1}=1$), \eqref{eq:perturbation} reads
\begin{equation}\label{eq:perturbation2}\begin{array}{ll}
    \lambda_A(\tilde{\bf{W}})&=1-\gamma\frac{1}{\frac{\beta M+\alpha}{\alpha}\frac{M^2+M\beta+M+\alpha}{M+\beta}} \frac{\beta}{\alpha}+o(\gamma)\\&=1-\gamma\frac{\beta(M+\beta)}{(\beta M+\alpha)(M^2+M\beta+M+\alpha)} +o(\gamma),
\end{array}\end{equation} 
yielding the claim.\qed
\end{pf}

Interestingly, by computing the relevant derivatives of \eqref{eq:linearA}, we can observe that increasing $\alpha$ leads to a larger first-order approximation of $\lambda_A$, suggesting that it slows down the convergence to the input. On the contrary, increasing $\beta$ decreases the approximation of $\lambda_A$, suggesting that this accelerates the consensus process. This is in contrast with what is observed for the scenario without input, where increase in $\alpha$ speeds up convergence to a consensus. These two contrasting insights provide analytical support to the numerical observations in  Figs.~\ref{fig:example4}--\ref{fig:example6}, which suggest the existence of a non-trivial tradeoff between the two mechanisms that leads to a non-monotonicity with respect to the two parameters.

In Fig.~\ref{fig:problem2}, we numerically evaluate this by leveraging Proposition~\ref{prop:convergence2}, and computing numerically the convergence rate in terms of the largest eigenvalue of $\mat{\tilde W}$. The results of our numerical computations suggest that the optimal speed of convergence is achieved with a tradeoff between the values of $\alpha$ and $\beta$, which should be both not too small neither too large. 

We believe that this tradeoff is due to the presence of the two contrasting mechanisms highlighted in our theoretical derivations. On the one hand, increasing $\alpha$ (and/or decreasing $\beta$) leads to faster convergence of the entire network to a common state (Corollary~\ref{cor:rate}). However, this slows down the speed at which information in an external input diffuses on the network (Proposition~\ref{prop:input}). This can be clearly seen in  Fig.~\ref{fig:example6}, in which most of the node trajectories quickly converge close to a consensus manifold, but then convergence to the input is extremely slow. On the other hand, decreasing $\alpha$ (and/or increasing $\beta$) leads to the opposite scenario, whereby the nodes in the unit with the input quickly reach consensus on the input, but the other trajectories lag (see Fig.~\ref{fig:example4}). This tradeoff further highlights how designing hierarchical organizations is a non-trivial problem, which calls for the use of system- and control-theoretic tools for their optimal design.

\begin{figure}
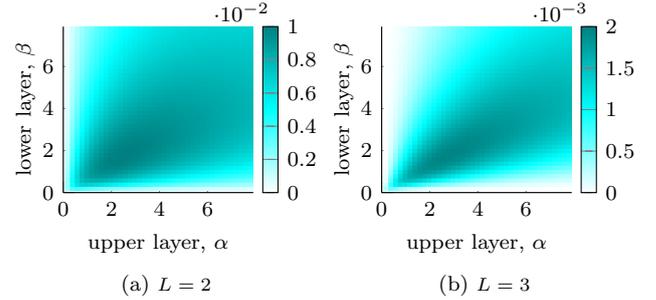

    \centering
  \centering
\subfloat[\scriptsize$L=2$]{\input{fig/5a}\label{fig:problem2_1}}
\subfloat[\scriptsize$L=3$]{\input{fig/5b}\label{fig:problem2_2}}
     \caption{Convergence rate $r$ of the bottom-up consensus problem, computed numerically (a) $L=2$ and b) $L=3$, for different values of $\alpha$ and $\beta$. In both scenarios $M=3$, $u=1$, and $\gamma_5=1$. }
    \label{fig:problem2}
\end{figure}

\section{Conclusion}\label{sec:conclusion}

In this paper, we studied the convergence rate of diffusive coupled single integrators on a hierarchical network structure. Our analysis elucidated how the structural characteristics of the hierarchy and the information processing dynamics impact the convergence rate. We found a crucial tradeoff between fast synchronization of the network and the ability to pass information from the lower layers upwards, which are both key features for the well-functioning of hierarchical organizations. 

The preliminary findings in this paper pave the way several  future research directions. First, our analytical results are limited to the 2-layer scenario. We aim to expand our analysis to networks with more layers,  to obtain a closed-form expression of the convergence rate for the bottom-up consensus problem, which can be used to optimally design a hierarchical organization. Second, this study focused networks without internal node dynamics. Our future research will consider nodes with non-trivial dynamics, which better captures a hierarchical organization in which nodes not only communicate but also perform tasks, e.g., solving a distributed optimization problem.


\bibliography{ifacconf}

\appendix
\section{Proof of Proposition~\ref{prop:spectrum_alphabeta}}\label{app:proof1} 

In this proof, we illustrate the derivation of all eigenvalues of $\mat W$ and the corresponding eigenspaces. First, it is known from the literature~\citep{Bullo2024}, but it is also 
straightforward to check that $\lambda_A=1$ is the dominant eigenvalue, has multiplicity equal to $1$, and its corresponding eigenvector is $\vect{v}=\vect{1}_{(M+1)^2}$.

Then, we consider the family of vectors defined as
\begin{equation}\label{eq:vbcd}
    \vect{\hat v}= \big[\vect{\hat v_1}^\top\,\, \vect{\hat v_2}^\top\,\,\dots\,\,\vect{\hat v_{M+1}}^\top\big]^\top,
\end{equation}
with
\begin{equation}
    \vect{\hat v_1}^\top=[\vect{y}^\top\,\,\,
     - \vect{y}^\top\vect{1}\,\,\,
     0],
\end{equation}
    with arbitrary $\vect{y}=[y_1,\dots,y_{M-1}]^\top\in\mathbb R^{M-1}$ with $\vect{y}\neq\vect{0}$, 
    \begin{equation}
    \vect{\hat v_\ell}^\top=[Ky_{\ell-1}\vect{1}_{M}^\top\,\,\,     Hy_{\ell-1}],
\end{equation}
   for any $\ell\in\{2,\dots,M\}$, and
    \begin{equation}
    \vect{\hat v_{M+1}}^\top=[ -K\vect{y}^\top\vect{1}\vect{1}_M^\top\,\,\,
     -H\vect{y}^\top\vect{1}]^,
\end{equation}
    where $H$ and $K$ are some constant parameters. Observe that for fixed the values of $H$ and $K$, $\vect{\hat v}$ constitutes an  $(M-1)$-dimensional subspace of $\mathbb R^n$, where the degrees of freedom are given by the number of independent entries of the vector $\vect{y}$ (i.e., the first $M-1$ entries of $\vect{\hat v}$), while all its other entries  are completely determined by $\vect{y}$ and the two constants $H$ and $K$. 
    
    Now, we illustrate that, for certain related values of $H$ and $K$, vector $\vect{\hat v}$ is an eigenvector of $\mat W$. In fact, by imposing the eigenvalue equation $\mat{W}\vect{\hat v}=\lambda\vect{\hat v}$, we obtain the conditions in \eqref{eq:lambdaB_general}, $H=\frac{K(M+\beta)}{K\beta-1}$, and also \eqref{eq:K_general}. Finally, letting $f(K)$ denote the left-hand side of \eqref{eq:K_general}, we observe that $f(K)=0$ has always three real and distinct solutions: $K_B>\frac{1}{\beta}$, $K_C\in(0,\frac1\beta)$, and $K_D<0$. This is because (as easily seen) $\lim_{K\to-\infty}f(K)=-\infty$, $f(0)=\alpha>0$, $f(\frac{1}{\beta})=-\frac{(\alpha+M)(\beta+M)}{\beta}<0$, and  $\lim_{K\to+\infty}f(K)=+\infty$, yielding the three real solutions. Then, to each of these distinct values of $K_{B,C,D}$, we associate $H_{B,C,D}=\frac{K_{B,C,D}(M+\beta)}{K_{B,C,D}\beta-1}$, and $\lambda_{B,C,D}$ following \eqref{eq:lambdaB_general}. Hence, each $\lambda_{B,C,D}$ with the subspace determined by the corresponding vector $\vect{\hat v}$ is an eigenvalue with multiplicity $M-1$ and corresponding $(M-1)$-dimensional eigenspace.

Then, we can verify that the $M(M-1)$  vectors 
\begin{equation}
    \vect{\check v}=e_{(M+1)i+1}-e_{(M+1)i+j}, 
\end{equation}
for $i\in\{1,\dots,M\}$ and $j\in\{2,\dots,M\}$ 
are all eigenvectors associated with eigenvalue $\lambda_E=\frac{\beta-1}{M+\beta}$. All these eigenvectors are linearly independent, forming a $M(M-1)$-dimensional eigenspace. Hence $\lambda_E$ is an eigenvalue of $\mat W$ with multiplicity $M(M-1)$.

Finally, we consider another family of vectors:
\begin{equation}\label{eq:vf}
    \vect{\bar v}=\begin{bmatrix}
    Ky\vect{1}_M^\top&
    K^2y&
    y\vect{1}_M^\top&
    Kv_1&       \dots&
  y\vect{1}_M^\top&
    Ky
    \end{bmatrix}^\top,
\end{equation}
for $y\neq 0$, parameterized by the constant $K$. We observe that, for specific values of such constant, $\vect{\bar v}$ is an eigenvalue of $\mat{W}$. In fact, imposing the eigenvalue equation $\mat W\vect{\bar v}=\lambda\vect{\bar v}$, we get the second-order equation: $(M+\alpha)(M+\beta)\lambda^2+((M+\alpha)(-\beta-M+1)-\alpha(M+\beta))\lambda-M+\alpha\beta+\alpha M-\alpha=0$, with $K=(M+\beta)(\lambda-1)+1$. 

One of its two solutions is $\lambda=1$, which forces $K=1$ and then $\bar {\bf{v}}=1$. this being a solution already identified earlier. Because the product of the two solutions of the equation is the ratio between the zeroth order and second order coefficients in the defining equation and one solution is 1, the remaining eigenvalue $\lambda_F$ is seen to be as as claimed. 

To conclude, the computation of $\lambda_{H,G}$ follows a similar approach to that used for obtaining $\lambda_{B,C,D}$. We consider the family of vectors
\begin{equation}\label{eq:vgh}
    \vect{\tilde v}=\begin{bmatrix}
    -\tfrac{M\beta-K\alpha\beta}{\alpha(K+\beta)}y\vect{1}_M^\top&
    -M\frac{\beta}{\alpha}y&
    y\vect{1}_M^\top&
    Kv_1&
       \dots&
 y\vect{1}_M^\top&
    Ky\\
    \end{bmatrix}^\top,
\end{equation}
for $y\neq 0$, with constant parameter $K$, and we observe that under some conditions on $K$, $\vect{\tilde v}$ is an eigenvector. These conditions, computed by imposing the eigenvalue equation, are $\lambda=\frac{M+\beta-1+K}{M+\beta}$ and the second-order equation $K^2+(2\beta+M-1)-+\beta(M+\beta-1)-\frac{(M+\beta)(M+\alpha\beta)}{M+\alpha}=0$. Finally, letting $K_{G,H}$ be the two solutions of the second-order equation (which are always real, as is guaranteed by the positivity of the discriminant), $\lambda_{G,H}$ are obtained by substituting them into the condition for $\lambda$. Observe that both these eigenvalues are simple, since once $K$ is defined, $\vect{\tilde v}$ determines a one-dimensional subspace of $\mathbb R^n$.

Finally, we observe that the total multiplicity of the eigenvalues computed in the above is equal to $1+3(M-1)+M(M-1)+1+2=M^2+2M+1=(M+1)^2$, which coincides with the dimension of the matrix $\mat W$, concluding the computation of the eigenvalues.
\qed
\end{document}